\documentclass[a4paper, 12pt]{article}

\usepackage[sort&compress]{natbib}
\bibpunct{(}{)}{;}{a}{}{,} 

\usepackage{amsthm, amsmath, amssymb, mathrsfs, multirow, url}
\usepackage[caption = false]{subfig}
\usepackage{graphicx} 
\usepackage{ifthen} 
\usepackage{amsfonts}
\usepackage[usenames]{color}
\usepackage{fullpage}
\usepackage[linesnumbered,algoruled,boxed,lined,commentsnumbered]{algorithm2e} 

\theoremstyle{plain} 
\newtheorem{thm}{Theorem}

\theoremstyle{definition}

\theoremstyle{remark} 
\newtheorem*{astep}{A-step}
\newtheorem*{pstep}{P-step}
\newtheorem*{cstep}{C-step}

\newcommand{\prob}{\mathsf{P}}

\newcommand{\unif}{{\sf Unif}}
\newcommand{\nm}{{\sf N}}

\newcommand{\gam}{{\sf Gamma}}

\newcommand{\chisq}{{\sf ChiSq}}

\newcommand{\weib}{{\sf Weib}}

\newcommand{\RR}{\mathbb{R}}

\newcommand{\YY}{\mathbb{Y}}
\newcommand{\UU}{\mathbb{U}}

\renewcommand{\S}{\mathcal{S}}

\newcommand{\hatG}{\widehat{G}}

\newcommand{\stleq}{\leq_{\text{st}}}
\newcommand{\stgeq}{\geq_{\text{st}}}

\title{Generalized inferential models for censored data}
\author{Joyce Cahoon\footnote{Department of Statistics, North Carolina State University; {\tt jyu21@ncsu.edu}, {\tt rgmarti3@ncsu.edu}} \quad \text{and} \quad Ryan Martin$^*$}
\date{\today}

\begin{document}

\maketitle 

\begin{abstract}   
Inferential challenges that arise when data are censored have been extensively studied under the classical frameworks.  In this paper, we provide an alternative generalized inferential model approach whose output is a data-dependent plausibility function.  This construction is driven by an association between the distribution of the relative likelihood function at the interest parameter and an unobserved auxiliary variable.  The plausibility function  emerges from the distribution of a suitably calibrated random set designed to predict that unobserved auxiliary variable.  The evaluation of this plausibility function requires a novel use of the classical Kaplan--Meier estimator to estimate the censoring rather than the event distribution.  We prove that the proposed method provides valid inference, at least approximately, and our real- and simulated-data examples demonstrate its superior performance compared to existing methods.  
\smallskip

\emph{Keywords and phrases:} Kaplan--Meier estimator; plausibility; random set; relative likelihood; survival analysis.
\end{abstract}

\section{Introduction}
\label{S:intro}  

Data are said to be {\em censored} when at least one of the observations is incomplete, i.e., only an interval that contains the actual value is available.  For example, in clinical trials or other time-to-event studies, it may happen that only a lower bound for the event time is observed because subjects drop out of the study, or the study ends before the event takes place.  This is called {right-censoring}.  Alternatively, in environmental applications, it may happen that only an upper bound on a chemical content is observed because the available device is limited to a certain detection level.  This is called {left-censoring}.  Of course, a combination of left- and right-censoring, or interval-censoring, is possible as well.  Beyond censoring direction, there are also Type~I and Type~II classifications, but we refer the reader to \citet{Klein2003} for these details.  For concreteness, we focus on Type~I right-censored data in a time-to-event setting, but it is easy to apply the same ideas for left- or interval-censored data and for contexts other than time.  

Let $X_i$ denote the event time and $C_i$ the censoring time for unit $i=1,\ldots,n$.  Under right censoring, the observed data consists of the pair 
\begin{equation}
T_i = \min(X_i, C_i), \quad D_i = 1(X_i \leq C_i), \quad i=1,\ldots,n, 
\label{eq:data}
\end{equation}
where $1(\cdot)$ is the indicator function, so that $D_i$ identifies whether $T_i$ is an event time or a censoring time.  Let $Y=\{(T_i, D_i): i=1,\ldots,n\}$ denote the observable data.  

A common assumption that we will adopt here is that of {\em random censoring}, where $X_1,\ldots,X_n$ are independent and identically distributed (iid) with continuous distribution function $F_\theta$, depending on a parameter $\theta \in \Theta$; $C_1,\ldots,C_n$ are iid with distribution function $G$; and the $X_i$'s and $C_i$'s are independent of one another \citep{lawless2011}.  Since the variables are time (or some other ``amount''), the statistical models, $F_\theta$, considered here and throughout the literature on this topic are supported on subsets of $(0,\infty)$ and are typically right-skewed.  The goal is to make inference on the unknown parameter $\theta$ of the time-to-event distribution; $G$ is an unknown nuisance parameter assumed to have no dependence whatsoever on $\theta$.  

For data $y = \{(t_i, d_i)\}$ observed from a random, Type~I, right-censored data generating process, \citet[][Sec.~3.5]{Klein2003} gives the likelihood function 
\begin{equation}
\label{eq:lik}
L_y(\theta) \propto \prod_{i=1}^n f_\theta(t_i)^{d_i} \bar F_\theta(t_i)^{1-d_i}, \quad \theta \in \Theta,
\end{equation}
where $f_\theta = F_\theta'$ and $\bar F_\theta = 1-F_\theta$ are the density and survival functions corresponding to $F_\theta$, respectively.  From the likelihood in \eqref{eq:lik}, it is relatively straightforward to produce point estimates, asymptotic confidence regions, or even Bayesian posterior distributions \citep{Ibrahim2001}.  These results, however, are not fully satisfactory as their coverage probabilities can be far from the target in finite samples.

In this paper, we take an alternative approach to construct an {\em inferential model} whose output takes the form of a non-additive, data-dependent belief/plausibility function.  This construction relies on a particular connection between the data, parameter, and an unobservable auxiliary variable.  Here, following the recommendations in \citet{plausfn, gim}, we make use of an association driven by the relative likelihood derived from \eqref{eq:lik}.  The belief function arises from the introduction of a (nested) random set aimed to predict that unobserved auxiliary variable.  An important consequence of this particular construction is that the belief function output inherits a calibration or {\em validity} property.  A precise statement is given in Section~\ref{S:background}, but an important practical consequence of the validity property is that the confidence, or plausibility, regions derived from the inferential model achieve the nominal frequentist coverage probability.  

Unfortunately, the presence of censoring complicates the basic inferential model construction and validity properties described in the references above.  In particular, the distribution $G$ of the censoring times is an infinite-dimensional nuisance parameter whose influence is difficult to overcome.  Here we propose an extension of the basic approach above, one that makes novel use of the Kaplan--Meier estimator \citep[e.g.,][]{Kaplan1958} for the censoring distribution $G$.  
From this, we develop a Monte Carlo algorithm to evaluate the belief and plausibility of any hypothesis about $\theta$, and we show---both theoretically and empirically---that inference drawn from the generalized inferential model output is valid, at least approximately, in the sense described in Section~\ref{S:background}.  Details of this construction and its properties are presented in Section~\ref{S:gim} and numerical examples comparing the proposed solution to that of more traditional methods are given in Section~\ref{S:examples}.  Finally, some concluding remarks are given in Section~\ref{S:discuss}.

\section{Background}
\label{S:background}

\subsection{Basic inferential models}

For observable data $Y \in \YY$, consider a statistical model $\{\prob_{Y|\theta}: \theta \in \Theta\}$ that contains candidate probability distributions for $Y$, indexed by a parameter space $\Theta$.  
As presented in \citet{imbasics, imbook}, an inferential model is a map from the available inputs, including observed data and posited statistical model, to a data-dependent function, $b_y: 2^\Theta \to [0,1]$, where $b_y(A)$ denotes the data analyst's degree of belief about the hypothesis $A \subseteq \Theta$ based on the observed data $Y=y$.  Naturally, inferences would be drawn from $b_y$.  This definition of a inferential model encompasses many different approaches, including those based on additive beliefs, e.g., Bayes, fiducial, and others, as well as non-additive beliefs like those discussed below.  

What properties should $b_y$ have?  
In the scientific applications we have in mind here, if it is desired that large $b_y(A)$ be interpreted as support for the claim that $A^c$ is false, then it becomes essential that the degrees of belief be calibrated so that we know what a ``large'' $b_y$ means, and consequently avoid making ``systematically misleading conclusions'' \citep{Reid2015}.  We formalize this need for an inferential model to be calibrated in terms of the following validity constraint: that $b_y$ satisfies 
\begin{equation}
\label{eq:valid}
\sup_{\theta \not\in A} \prob_{Y|\theta}\{ b_Y(A) > 1-\alpha\} \leq \alpha, \quad \forall \; \alpha \in [0,1], \quad \forall \; A \subseteq \Theta. 
\end{equation}
That is, if the hypothesis $A$ is false, so that $A \not\ni \theta$, the degree of belief $b_Y(A)$, as a function of $Y \sim \prob_{Y|\theta}$, will be stochastically no larger than $\unif(0,1)$.  This validity condition can equivalently be expressed in terms of the plausibility function, $p_y(A) = 1-b_y(A^c)$, the belief function's dual \citep{Shafer1976}.  This dual inferential model output is valid if 
\begin{equation}
\label{eq:valid2}
\sup_{\theta \in A} \prob_{Y|\theta}\{ p_Y(A) \leq \alpha\} \leq \alpha, \quad \forall \; \alpha \in [0,1], \quad \forall \; A \subseteq \Theta.
\end{equation}
Following this constraint, the plausibility values can be compared to a $\unif(0,1)$ scale, and decisions based on such comparisons will control frequentist error rates \citep{gim}.  

Based on the {\em false confidence theorem} in \citet{balch2019}, \citet{martin.nonadditive} argues that validity as in \eqref{eq:valid} requires that the degrees of belief be non-additive.  Since we take this validity property to be fundamental to the logic of statistical inference, we focus here on genuinely non-additive degrees of belief, e.g., the belief/plausibility functions in \citet{Shafer1976} or the special case of necessity/possibility functions in \citet{Dubois2006}, \citet{Dubois2012}, and \citet{destercke.dubois.2014}.

How to construct a valid inferential model?  The original construction in \citet{imbasics}, starts with an association, i.e., a characterization of the statistical model based on what is called an auxiliary variable.  The prototype for this takes the form 
\begin{equation}
\label{eq:basic}
Y = a(\theta, U), \quad U \sim \prob_U, 
\end{equation}
where $a$ is a given function and $\prob_U$ is a distribution for $U \in \UU$ that does not depend on any unknown parameters.  This describes an algorithm for simulating from $\prob_{Y|\theta}$ but also guides our intuition about inference.  That is, {\em if $U$ were observable, along with $Y$, then the best possible inference follows by simply solving \eqref{eq:basic} for $\theta$}, as in \eqref{eq:basic.focal}.  Since $U$ is actually unobservable, it is tempting to create a sort of ``posterior distribution'' for $\theta$ by taking draws from $\prob_U$, plugging them into \eqref{eq:basic}, with the observed $Y=y$, and solving for $\theta$.  This is basically Fisher's fiducial argument \citep[e.g.,][]{Fisher1973, Dempster1963, Hannig2016}, which generally leads to additive beliefs that fail to meet the validity condition.  Non-additivity can be introduced by stretching points sampled from $\prob_U$ into random sets designed to hit the unobserved value of $U$ in \eqref{eq:basic} that corresponds to the observed $Y=y$ and the true value of $\theta$.  The following three steps summarize this construction.

\begin{astep}
Given the {\em association} \eqref{eq:basic} and the observed $Y=y$, define the focal elements
\begin{equation}
\label{eq:basic.focal}
\Theta_y(u) = \{\theta: y = a(\theta, u)\}, \quad u \in \UU. 
\end{equation}
\end{astep}

\begin{pstep}
Introduce a random set $\S \sim \prob_\S$, taking values in $2^\UU$, designed to {\em predict} the unobserved value of $U$ in \eqref{eq:basic}. 
\end{pstep}

\begin{cstep}
{\em Combine} the output of the A- and P-steps to get a new random set 
\[ \Theta_y(\S) = \bigcup_{u \in \S} \Theta_y(u), \quad \S \sim \prob_\S, \]
and define the belief function,  
\[ b_y(A) = \prob_\S\{\Theta_y(\S) \subseteq A\}, \quad A \subseteq \Theta, \]
and its dual, the plausibility function, $p_y(A) = 1-b_y(A^c)$.  
\end{cstep}
Under very mild conditions on the user-specified random set $\S$, the corresponding inferential model is valid in the sense of \eqref{eq:valid}.  Indeed, the only requirement is that $\S$ be calibrated to predict unobserved draws from $\prob_U$.  This is relatively easy to arrange because $\prob_U$ is known and $\S \sim \prob_\S$ is user-specified.  More specifically, let $\gamma(u) = \prob_\S(\S \ni u)$, an ordinary function on $\UU$, be determined implicitly by $\prob_\S$; note that $\gamma$ is the plausibility contour corresponding to $\S$.  Then validity as in \eqref{eq:valid} corresponds to a stochastic dominance property, namely, $\gamma(U) \stgeq \unif(0,1)$.  For example, in what follows, we work with a random set $\S$ of the form 
\begin{equation}
\label{eq:one.sided}
\S = [\tilde U, 1], \quad \tilde U \sim \prob_U := \unif(0,1),
\end{equation}
so that $\gamma(u) = u$ and, hence, $\gamma(U) = U \sim \unif(0,1)$.  Though not strictly necessary for validity, efficiency considerations suggest that $\S$ be nested, like in \eqref{eq:one.sided}, which makes the belief function consonant; the validity property together with consonance is reminiscent of the confidence structure developments in \citet{balch2012}.  

\subsection{Generalized inferential models}
\label{S:gimintro}

As \citet{gim} argued, the above formulation can be rather rigid; greater flexibility and, in some cases, improved performance can be gained by working with a so-called {\em generalized association}, one that does not fully characterize the posited statistical model.  As above, suppose we have data $Y \sim \prob_{Y|\theta}$, but consider 
\begin{equation}
\label{eq:gassoc}
R_Y(\theta) = H_\theta^{-1}(U), \quad U \sim \prob_U = \unif(0,1), 
\end{equation}
where $R_y(\theta)$ is some real-valued function of our parameter of interest $\theta$, indexed by the data, and $H_\theta$ is its distribution function,  
\[ H_\theta(r) = \prob_{Y|\theta}\{R_Y(\theta) \leq r\}, \quad r \in \RR. \]
Unlike \eqref{eq:basic}, the relation \eqref{eq:gassoc} does not describe the data-generation process, it only establishes a link between data, parameter, and auxiliary variable, which is all that was needed for the inferential model construction described above.  

The advantage of this generalized association is, as explained in \citet{gim}, that we have directly reduced the dimension of the auxiliary variable, from at least the dimension of $\theta$ down to 1.  This greatly simplifies the construction of a (good) random set $\S$ for predicting that unobservable quantity.  What is an appropriate choice of $R_y(\theta)$?  The options are virtually unlimited, but since dimension reduction would generally result in loss of information, and since we prefer to retain as much information as possible, we opt to take $R_y(\theta)$ as the {\em relative likelihood}
\begin{equation}
\label{eq:rel.lik}
R_y(\theta) = L_y(\theta) / L_y(\hat\theta),
\end{equation}
where $\hat\theta$ is the maximum likelihood estimator, i.e., $\hat\theta = \arg\max_\vartheta L_y(\vartheta)$.  Extensive studies have explored the use of relative likelihood to define degrees of belief \citep[e.g.,][]{Shafer1976, Wasserman1990}, but they focus on examples where the likelihood cannot be normalized or where a normalized likelihood is misleading \citep{Shafer1982}.  Our approach differs in the sense that we can evaluate the distribution of the relative likelihood by Monte Carlo.  From here, the inferential model construction is conceptually straightforward.

\begin{astep}
Set $\Theta_y(u) = \{\theta: R_y(\theta) = H_\theta^{-1}(u)\}$ for $u \in [0,1]$.  
\end{astep}

\begin{pstep}
Define $\S = [\tilde U, 1]$, where $\tilde U \sim \unif(0,1)$ like in \eqref{eq:one.sided}; so that the distribution, $\prob_\S$, is fully determined by the uniform distribution.
\end{pstep}

\begin{cstep}
Combine the two sets above to get 
\[ \Theta_y(\S) = \bigcup_{u \in \S} \Theta_y(u) = \{\theta: H_\theta(R_y(\theta)) \geq \tilde U\}, \quad \tilde U \sim \unif(0,1). \]
Then the plausibility contour is 
\begin{equation}
\label{eq:pl.contour}
p_y(\theta) := \prob_\S\{\Theta_y(\S) \ni \theta\} = H_\theta(R_y(\theta)), \quad \theta \in \Theta, 
\end{equation}
which determines the full belief and plausibility functions.  It follows from Theorem~1 in \citet{gim} that the generalized inferential model with plausibility function determined by \eqref{eq:pl.contour} achieves the validity property in \eqref{eq:valid2}.  
\end{cstep}


It is often the case that the full parameter of the statistical model is of the form $(\theta, \eta)$, i.e., $Y \sim \prob_{Y|\theta,\eta}$, where $\theta$ is the quantity of interest and $\eta$ is a so-called {\em nuisance parameter}.  The censored data application considered here is of this form---with the censoring distribution $G$ being the nuisance parameter---as is the meta-analysis application in \citet{immeta}.  A very natural way to proceed with marginal inference on $\theta$, which we describe in Section~\ref{S:gim}, is to define a function $R_Y(\theta)$ that does not directly depend on the value of the nuisance parameter $\eta$.  This does not immediately resolve the $\eta$-dependence, however, because the distribution function 
\begin{equation}
\label{eq:H}
r \mapsto H_{\theta,\eta}(r) := \prob_{Y|\theta,\eta}\{R_Y(\theta) \leq r\} 
\end{equation}
will generally depend on the unknown $\eta$.  To overcome this dependence on the unknown nuisance parameter, one might try plugging in an estimator $\hat\eta$ based on the available data, which amounts to constructing a generalized inferential model based on the approximate distribution function for  $R_Y(\theta)$, namely, $H_{\theta, \hat\eta}$.  Of course, plugging in an estimate affects the exact validity of the generalized inferential model but, at least intuitively, if $\hat\eta$ is a reasonably accurate estimate of $\eta$, then the corresponding plug-in generalized inferential model ought to be approximately valid.  This is precisely the situation encountered in censored-data applications, and Theorem~\ref{thm:valid} below confirms the above intuition.

\section{Generalized inferential models under censoring}
\label{S:gim}

\subsection{Construction}

For random right-censored data, the full likelihood for $Y$, where $Y_i$ are independently generated from \eqref{eq:data}, is given by 
\[L_y(\theta) = \prod_{i=1}^n \bar{G}(t_i)^{d_i} g(t_i)^{1-d_i} \prod_{i=1}^n f_\theta(t_i)^{d_i} \bar{F}_\theta(t_i)^{1-d_i}.\]
Since our interest is only in the $\theta$ parameter and our censoring times do not depend on $\theta$, the first term of the full likelihood is treated as constant.  Therefore, when we construct the relative likelihood $R_Y(\theta)$ as in \eqref{eq:rel.lik}, we effectively eliminate the nuisance parameter $G$. 

Our generalized inferential model for censored data thus proceeds as outlined in Section~\ref{S:gimintro}, in which the relative likelihood is the connection between the data $Y$, our interest parameter $\theta$, and a scalar auxiliary variable $U$.  That is, we have $R_Y(\theta) = H_{\theta, G}^{-1}(U)$, where $U \sim \prob_U = \unif(0,1)$, where $H_{\theta,G}$ is the distribution function of the relative likelihood as in \eqref{eq:H}; note that, while the relative likelihood itself does not depend on $G$, its distribution does.  That completes the A-step of the construction.  Given the form of the relative likelihood, values closer to 1 suggest values of $\theta$ that are more likely so, for the P-step, we choose predictive random sets in the form of nested intervals $\S = [\tilde{U}, 1]$, where $\tilde{U}\sim\unif(0,1)$, to predict the one-dimensional auxiliary variable $U$.  Then the C-step proceeds exactly as in Section~\ref{S:gimintro}.  If $G$ were known, then it would be relatively simple to evaluate the distribution function $H_{\theta,G}$ and, hence, the plausibility contour in \eqref{eq:pl.contour}; moreover, validity of the generalized inferential model would follow immediately from the general theory.  Of course, $G$ is never known in applications, so we need to suitably modify the above strategy in order to overcome this challenge.  As we indicated above, it makes sense to plug in an estimator of $G$, but the construction of an estimator and justification of the corresponding plug-in method are non-trivial.  The next two subsections address these challenges in turn.

\subsection{Implementation}
\label{S:implementation}

Putting the above inferential model construction into practice requires that the distribution function of the relative likelihood be evaluated, at least approximately, for every $\theta$.  This is straightforward to do when data are not censored.  This is similarly straightforward if data are censored but the censoring distribution $G$ is known.  Indeed, a simple Monte Carlo approximation is available:
\begin{equation}
\label{eq:mc}
H_{\theta, G}(r) \approx \frac{1}{M} \sum_{m=1}^M 1\{R_{Y^{(m)}} (\theta) \leq r \},
\end{equation}
where $\{Y^{(m)}: m = 1, \ldots, M\}$ are independent copies of $Y^*=\{(T_i^*, D_i^*): i=1,\ldots,n\}$ and $(T_i^*, D_i^*)$ as in \eqref{eq:data}, with $X_i^*$ iid from $F_\theta$ and $C_i^*$ iid from the known censoring distribution $G$.  However, in our present context, $H_{\theta, G}$ depends (implicitly) on the unknown distribution $G$ of censoring times, so something more sophisticated than that simple strategy just described is needed.  Here we recommend using a plug-in estimator $\hatG$.   

The Kaplan--Meier estimator was not originally designed to estimate the censoring distribution function, but it is straightforward to simply reverse the event/censored classification.  That is, we still observe $T_i = \min(X_i, C_i)$ but now we think of $C_i$ as the ``event time'' and $X_i$ is the ``censoring time.''  Then we construct the Kaplan--Meier estimator $\hatG$ based on this alternative perspective.  

After swapping the observed/censored classifications, obtaining the Kaplan--Meier estimate is straightforward; we use the built-in functions in R's {\tt survival} package \citep{Therneau2015}.  But there are a few technical points worth making about the estimation process.  Recall that, in typical applications of the Kaplan--Meier estimator of a survival function $S(t)$, if the largest observation corresponds to a ``censored'' outcome, then $\hat S(t)$ {\em does not} vanish as $t \to \infty$, which amounts to putting some positive amount of mass at $\infty$.  In our context, since we interpret the original event times as censored times, our estimate $\hatG$ will put positive mass at $\infty$ when the largest observation is an event, under which $C_i^*$'s drawn from $\hatG$ will equal $\infty$ and, consequently, $T_i^*$'s drawn will correspond to an event time as $X_i^* < C_i^*$.  Our numerical simulations suggest that the Monte Carlo samples, $Y^*$, attained in this manner, reflect the censoring level in the original data.

\subsection{Validity properties}
\label{SS:valid}

That the corresponding inferential model satisfies the validity property follows immediately from the arguments presented in \citet{gim}.  Since our predictive random sets are tailored such that the plausibility contours are stochastically no larger than uniform, i.e., $H_{\theta,G}(R_Y(\theta)) \stleq \unif(0,1)$ when $Y \sim \prob_{Y|\theta,G}$, then
\begin{equation}
\label{eq:validforcensored}
\sup_{\theta \in A} \prob_{Y|\theta, G}\{ p_Y(A) \leq \alpha\} \leq \alpha, \quad \forall \; \alpha \in [0,1], \quad \forall \; A \subseteq \Theta. 
\end{equation}
A desirable consequence of validity is that confidence regions having the nominal frequentist coverage probability can be constructed immediately based on the plausibility function output.  Indeed, the set 
\begin{equation}
\label{eq:interval}
\{\theta: p_y(\theta) > \alpha\}
\end{equation}
is a nominal $100(1-\alpha)$\% confidence region for any $\alpha \in (0,1)$.  This follows since the probability that the above region contains the true parameter value $\theta$ equals the probability that $p_Y(\theta) > \alpha$ which, in turn, equals $1-\alpha$. 

Can anything be said about validity of the inferential model derived from the above algorithm with the plug-in estimator $\hatG$?  That is, can we conclude that 
\[ \prob_{Y|\theta, G}\{p_Y(\theta; \hatG) \leq \alpha\} \leq \alpha, \]
at least approximately?  Here $p_y(\theta; \hatG)$ denotes the plausibility function obtained by applying the above algorithm with $\hatG$ plugged in for the unknown $G$, i.e., simulating $C_i^*$'s iid from $\hatG$.  The dependence of $p_y(\theta; \hatG)$ on the Kaplan--Meier estimator, an infinite-dimensional quantity, is quite complicated, but at the very least, under mild assumptions, our proposed generalized inferential model should be valid for large $n$.  The following theorem confirms this.  Since we are considering asymptotic properties as $n \to \infty$, we embellish on our previous notation to emphasize the dependence on $n$.   

\begin{thm}
\label{thm:valid}
Let $Y^n = (Y_1,\ldots,Y_n)$ be a sample obtained under random censoring, $\prob_{Y^n|\theta,G}$, as described in \eqref{eq:data}, where both $\theta$ and $G$ are unknown.  The proposed plausibility function for inference on $\theta$, defined by
\[p_{Y^n} (\theta; \hatG_n) = H_{\theta, \hatG_n}^n (R_{Y^n} (\theta)), \]
with $\hatG_n$ the Kaplan--Meier estimator of $G$ described above based on $Y^n$, satisfies 
\[ p_{Y^n}(\theta; \hatG_n) \to \unif(0,1) \quad \text{in distribution as $n \to \infty$}. \] 
\end{thm}

\begin{proof} 
See Appendix~\ref{SS:proof2}.
\end{proof}

Theorem~\ref{thm:valid} establishes approximate validity, but the theoretical support for our approach might actually be stronger than the theorem suggests.  First, the proof amounts to comparing the distribution of $R_{Y^n}(\theta)$ under two different distributions for $Y^n$, one based on $(\theta,G)$ and one on $(\theta,\hatG_n)$.  The relative likelihood is well-known to be an approximate pivot, and its distributional dependence on these parameters is more sensitive in the $\theta$ direction than in the $G$ direction.  So, since we are holding $\theta$ fixed and only moving a small amount in the $G$ direction, we can expect that our approximation would be quite accurate.  To see this accuracy in action, we take 10,000 samples of size $n=15$ in which $X_i$'s are generated from a standard exponential subject to random right censoring from the $\unif(0,5)$.  A Monte Carlo estimate of the distribution function $\alpha \mapsto \prob_{Y|\theta, G}\{p_Y(\theta; \hatG) \leq \alpha\}$ shown in Figure~\ref{fig:validity} is approximately uniform, hence approximate validity.  Moreover, by starting with the relative likelihood $R_y(\theta)$ in \eqref{eq:rel.lik}, we removed almost all dependence on the nuisance parameter $G$; that is, the exact distribution of our relative likelihood is roughly constant in $G$ and thus the plug-in estimator we used to get $\hatG$ apparently does not need to be especially accurate.  As a result, the plausibility output using our plug-in method as described in Section~\ref{S:implementation} is close to the exact distribution.  Simulated- and real-data examples in Section~\ref{S:examples} further demonstrate the proposed method's strong performance compared to others, supporting our validity claim.  

\begin{figure}[t]
\begin{center}
  \includegraphics[width=.5\textwidth]{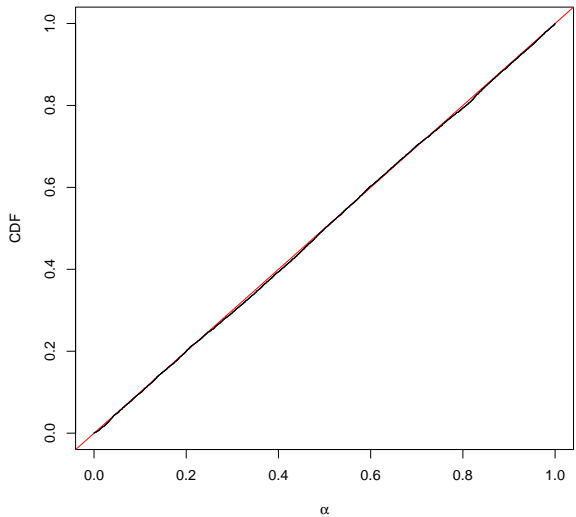}
  \caption{Distribution of $\alpha \mapsto \prob_{Y|\theta, G}\{p_Y(\theta; \hatG) \leq \alpha\}$ (black) compared with that of $\unif(0,1)$ (red) based on Monte Carlo samples from a standard exponential distribution subject to random right censoring. The average censoring level among all 10,000 replications at this setting is 19.9\%.}
  \label{fig:validity}
  \end{center}
\end{figure}

\section{Examples}
\label{S:examples}

We compare our proposed approach against frequentist and Bayesian methods with simulated and real data.  The exponential and Weibull examples are taken from the {\tt survival} package in R, while the last log-normal example is taken from \citet{Krish2011}.  We consider these three  parametric distributions that are commonly used in time-to-event analyses, and we generate 10,000 replications of censored data under various settings of these distributions.  We repeat each set of simulations at four sample sizes of $n \in \{15, 20, 25, 50\}$.  As our results suggest, plausibility functions consistently outperform more familiar methods, achieving nearly the nominal $100(1-\alpha)$\% coverage rate across different distributions, parameter settings, and sample sizes. 

\subsection{Exponential}

The classic time-to-event distribution is exponential, characterized by a constant hazard rate $\theta > 0$, in which the density function is $f_\theta(t) = \theta e^{-\theta t}$.  For $n$ items, independently subject to random right censoring, summarized by $y=\{(t_i,d_i)\}$ as above, the maximum likelihood estimate is $\hat{\theta} = \sum_{i=1}^n d_i / \sum_{i=1}^n t_i$.  From its asymptotic normality, a 95\% confidence interval is easily obtained as $\hat{\theta} \pm 1.96 I(\hat{\theta})^{-1/2}$, where $I(\hat{\theta})$ is the observed information.  From a Bayesian standpoint, the censoring mechanism can be safely ignored as the likelihood can be formed from \eqref{eq:lik} and combined with a conjugate $\gam(\alpha_0, \beta_0)$ prior to arrive at the posterior $\gam(\alpha_0 + \sum_{i=1}^n d_i, \beta_0 + \sum_{i=1}^n t_i)$.  Posterior credible intervals are then easily obtained.  Experiments with various values of $(\alpha_0,\beta_0)$ revealed that $\alpha_0=2$ and $\beta_0=1$ had the best overall performance across our settings with respect to coverage probability of the credible intervals.  

From an inferential model perspective, we begin with the baseline association of the relative likelihood for $\theta \in \Theta$,
\begin{equation}
\frac{\theta^{\sum_i D_i} e^{-\theta \sum_i T_i}}{\hat\theta^{\sum_i D_i} e^{-\hat\theta \sum_i T_i}} = H_{\theta, G} ^{-1}(U), \quad U\sim \prob_U = \unif(0,1).
\label{eq:imexp}
\end{equation}
As described above, we write $R_Y(\theta)$ for the left-hand side of the above display.  For fixed data $y$, we follow through our A-step with the singleton-valued map
\[ \Theta_y(u) =\{\theta: R_y(\theta) = H_{\theta, G} ^{-1}(u) \}, \quad u \in [0,1]. \]
Next, the P-step requires introducing a predictive random set $\S$ in \eqref{eq:one.sided} for $U$.  We then combine our A- and P-steps
\[ \Theta_y(\S) = \bigcup_{u \in \S} \Theta_y(u) = \{\theta: H_{\theta, G}(R_y(\theta)) \geq \tilde U\}, \quad \tilde U \sim \unif(0,1). \]
And we summarize the distribution of this random set $\Theta_y(\S)$ by a plausibility function
\[ p_y(\theta) = H_{\theta, G}(R_y(\theta)), \quad \theta > 0. \]
A 100($1-\alpha$)\% confidence interval can be obtained as the upper level set of the plausibility function as in \eqref{eq:interval}.  Evaluating this plausibility function requires the Monte Carlo procedure discussed in Section~\ref{S:gim}.  

\begin{figure}[p]
   \centering
   \subfloat[][$n=15$]{\includegraphics[width=.45\textwidth]{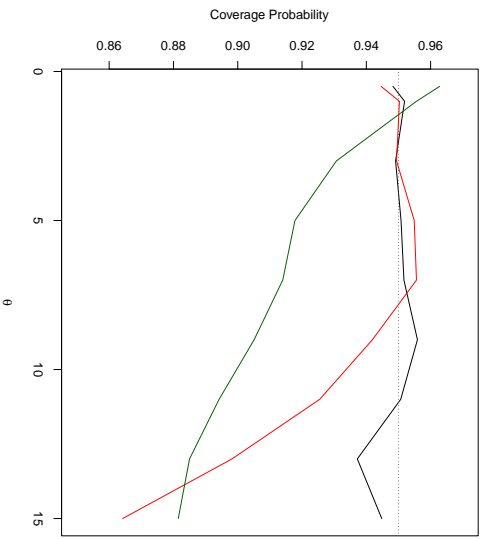}}\quad
   \subfloat[][$n=20$]{\includegraphics[width=.45\textwidth]{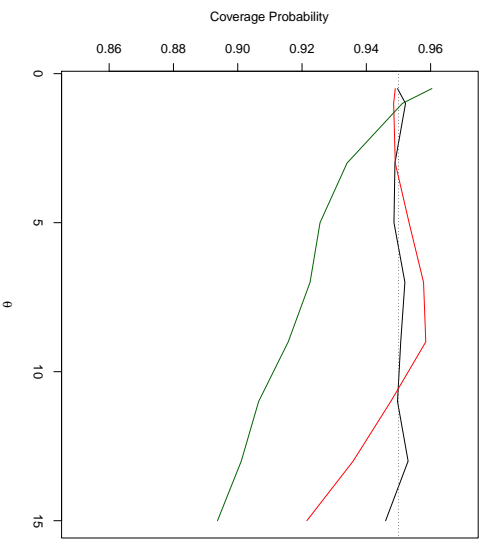}}\quad
   \subfloat[][$n=25$]{\includegraphics[width=.45\textwidth]{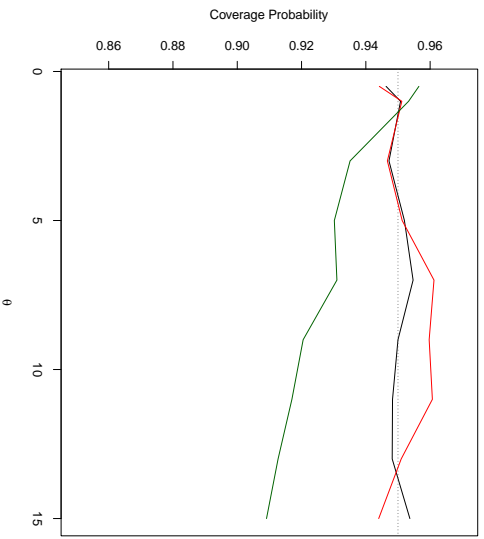}}\quad
   \subfloat[][$n=50$]{\includegraphics[width=.45\textwidth]{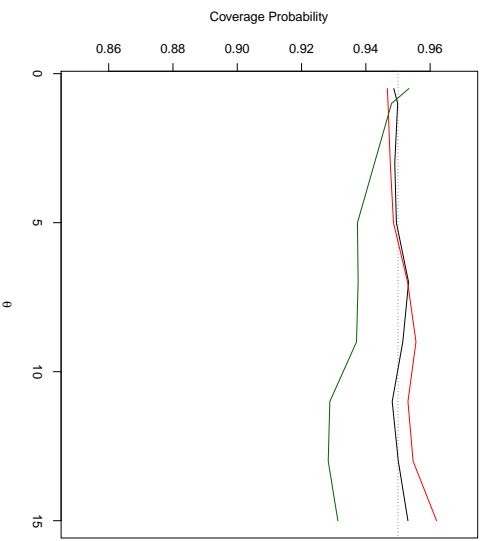}}\quad
   \caption{Coverage probability of the 95\% plausibility region for $\theta$ in the exponential model (black). Results compared to those based on maximum likelihood (red) and Bayesian with a $\gam(2,1)$ prior (green).}
   \label{fig:expsim}
\end{figure}

For comparison, we simulate 10,000 replications of lifetimes arising from nine different $\theta$ settings in the exponential distribution.  For each of these 90,000 simulations, the lifetimes $X_1, \ldots, X_n \sim F_\theta$ generated were subject to random right censoring from $C_1, \ldots, C_n \sim \unif(0,5)$, allowing us to compare the coverage of our inference procedure under a wide range of censoring levels.  Results shown in Figure \ref{fig:expsim} demonstrate that the nominal 100($1-\alpha$)\% coverage is achieved by our proposed method.  Note that this problem is particularly challenging in the $n=15$ and large $\theta$ case, since large $\theta$ implies more censoring.  The maximum likelihood and Bayes approaches appear to be substantially affected by this extreme censoring, while our generalized inferential model is not.  


For a real-data illustration, we consider the primary biliary cirrhosis (PBC) data from a clinical trial at the Mayo Clinic from 1974 to 1984.  The data consists of $n=312$ recorded survival times for patients involved in the randomized trial, along with a corresponding right censoring indicator; there are 168 censored cases, more than 50\% of total observations.  Figure~\ref{fig:pbc} shows the point plausibility function $p_y(\theta)$ for a range of parameter values, along with the corresponding 95\% plausibility interval \eqref{eq:interval}.  For comparison, 95\% confidence intervals based on asymptotic normality of the maximum likelihood estimate are also displayed.  The intervals derived from the plausibility function are almost indistinguishable from the likelihood-based intervals, which is a sign of our proposed approach's efficiency, since the latter are the asymptotically ``best'' intervals.  

\begin{figure}[t]
\begin{center}
\includegraphics[width=.75\textwidth]{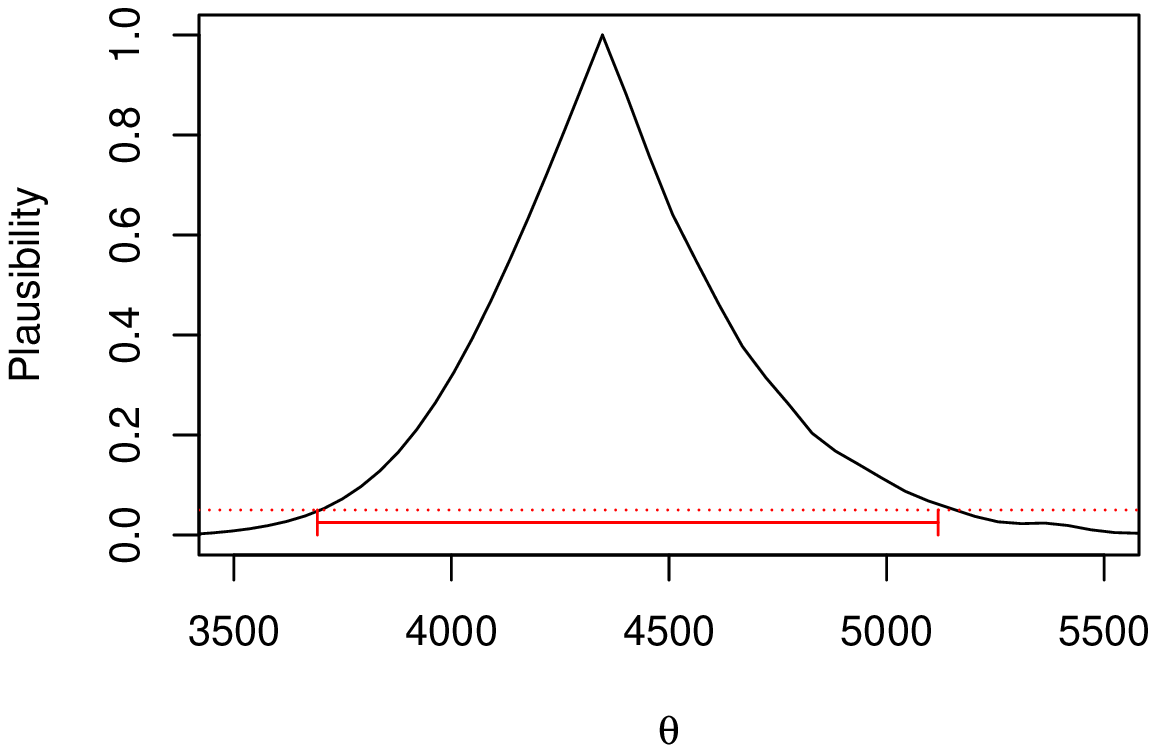}
\caption{Point plausibility function for the mean in the PBC example under an exponential model (black). Reference line at $\alpha = 0.05$ (dotted) and approximate 95\% confidence intervals based on maximum likelihood (red).}
\label{fig:pbc}
\end{center}
\end{figure}


\subsection{Weibull}

One of the most widely used time-to-event distributions is the Weibull, with applications in manufacturing, health, etc., as it has sufficient flexibility to capture changes in the hazard rate \citep{lawless2011}.  Exponential is a special case of the Weibull when the shape parameter $\beta = 1$.  The density and survival functions, indexed by $\theta = (\beta, \lambda)$, are
\[ f_\theta(t)= \lambda \beta t^{\beta-1}\exp{(-\lambda t^\beta)}, \quad \bar{F}_\theta(t) = \exp{(-\lambda t^\beta)}. \]
Similar to the setup as described for the exponential example, we compare the performance of our proposed approach against that of a more traditional frequentist or objective Bayesian approach.  An inferential model requires that we simulate the distribution of $R_Y( \theta)$; so for a finite grid of $\theta=(\beta,\lambda)$ values, for each pair, 500 Monte Carlo samples of $Y^*$ are obtained by taking the minimum between realizations of $X^* \sim \weib(\beta, \lambda)$ and $C^* \sim \hatG$, the modified Kaplan--Meier estimate.  We implement this procedure for 10,000 replications of lifetimes arising from six different settings of the Weibull distribution.  These 60,000 replications were each subject to random right censoring from $G \sim \unif(0,4)$.  For a Bayes approach, multiple non-informative and weakly informative priors were used, from which the $\gam(0.1, 1)$ prior on the shape and $\nm(0,10)$ prior on the log transformed scale were selected, as they resulted in credible regions with the highest coverage.  Surprisingly, as shown in Figure~\ref{fig:weisim}, despite the careful specification of these priors, the generalized inferential model still remains the only method that achieves nominal coverage across these censored data settings.  The joint confidence sets from maximum likelihood under-cover, while the joint credible sets from our Bayes approach achieves nominal coverage only past a specific censoring threshold.  Further investigations into interval lengths (not shown) also demonstrate our plausibility intervals are on average shorter than the Bayesian intervals.


\begin{figure}[p]
   \centering
   \subfloat[][$n=15$]{\includegraphics[width=.45\textwidth]{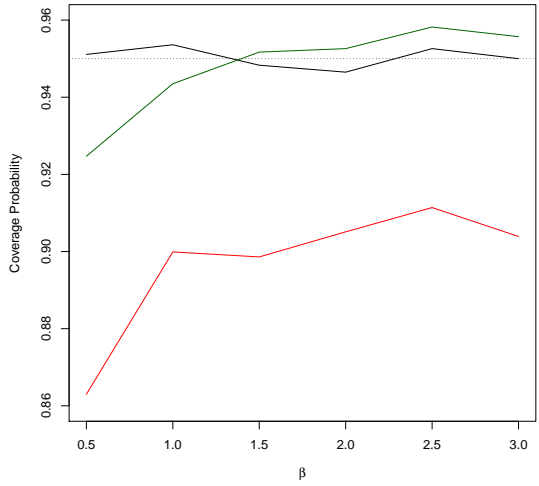}}\quad
   \subfloat[][$n=20$]{\includegraphics[width=.45\textwidth]{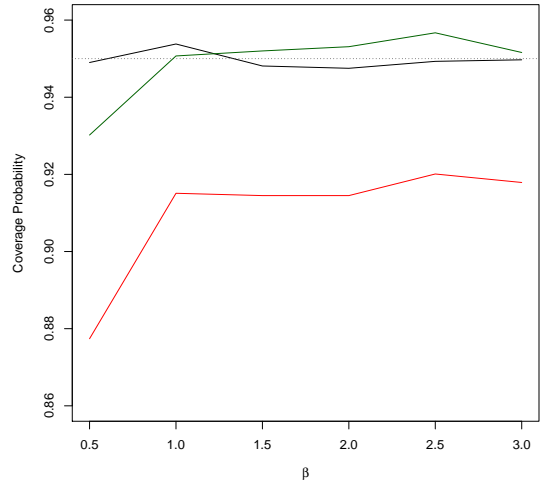}}\quad
   \subfloat[][$n=25$]{\includegraphics[width=.45\textwidth]{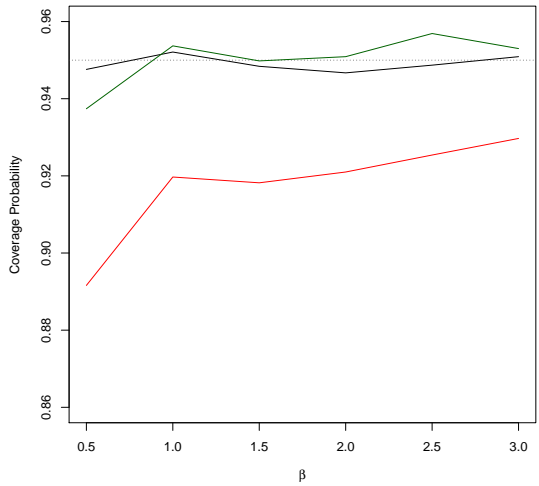}}\quad
   \subfloat[][$n=50$]{\includegraphics[width=.45\textwidth]{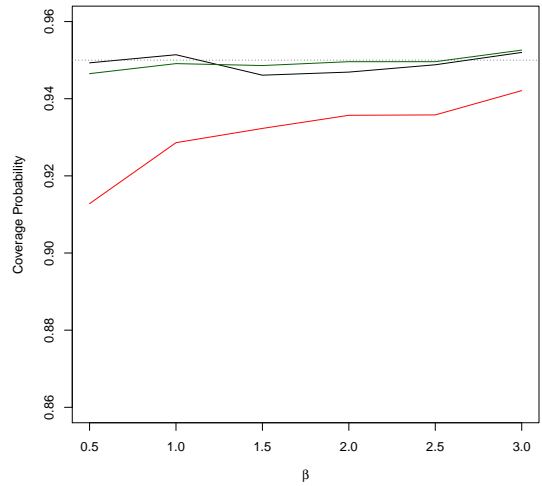}}\quad
   \caption{Coverage probability of the 95\% plausibility region for $\theta=(\beta, \lambda)$ in the Weibull model (black). Results compared to  maximum likelihood (red) and Bayesian intervals based on a $\gam(0.1,1)$ prior on the shape and $\nm(0, 10)$ prior on the log transformed scale (green).}
   \label{fig:weisim}
\end{figure}

For a real-data example, we consider survival data on ovarian cancer patients from a clinical trial that took place from 1974 to 1977.  This data set has $n = 26$ survival times for patients that entered the study with stage II or IIIA cancer and were treated with cyclophosphamide alone or cyclophosphamide with adriamycin.  Of this patient group, 14 survived (or was censored) by the end of the study, while 12 died \citep{Edmonson1979}.  Despite the small sample size and high censoring level, our plausibility contours capture the non-elliptical shape as shown by the Bayesian posterior in Figure~\ref{fig:ovar}.

\begin{figure}[t]
\begin{center}
\includegraphics[width=.75\textwidth]{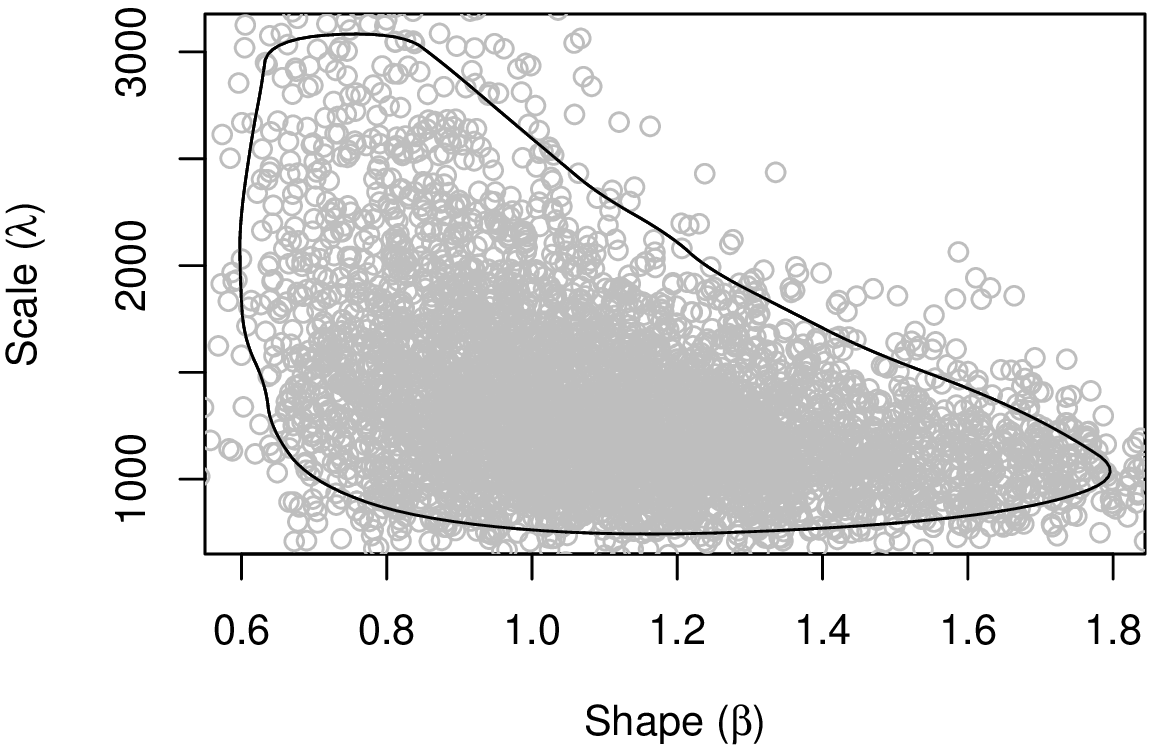}
\caption{Plausibility contour (black) for $\theta=(\beta,\lambda)$, the shape and scale parameter pair, in the ovarian cancer data under a Weibull model subject to Type~I right censoring. Bayesian posterior samples based on a $\gam(1, 0.1)$ prior for the shape and $\nm(0, 10)$ prior for the log transformed scale parameter (gray).}
\label{fig:ovar}
\end{center}
\end{figure}

\subsection{Log-normal}

Within environmental science, the log-normal distribution is often used to approximate data that are censored to the left, e.g., chemical pollutants that can only be detected above some minimal threshold \citep{Krish2011}. The density function, indexed by $\theta=(\mu, \sigma)$, is 
\[ f_\theta(t) = \frac{1}{(2\pi)^{1/2} \sigma t} \exp\Bigl\{-\frac{1}{2} \Bigl(\frac{\log t-\mu}{\sigma}\Bigr)^2\Bigr\}. \]
Similar to our examples above, we compare the coverage performance of our plausibility contours against that of ellipses based on asymptotic normality of the maximum likelihood estimator and posterior credible regions based on a $\gam(1, 0.1)$ prior on the precision $\tau = \sigma^{-2}$ and $\nm(0, 1000/\tau)$ prior on the mean.  Again, 10,000 replications of censored data were generated from 6 different settings of the log-normal distribution, subject to left censoring under $G \sim \unif(0,1)$.  In order to approximate the distribution of $R_Y(\theta)$, however, our modified Kaplan--Meier estimate $\hatG$ now requires putting positive mass at 0 when the smallest observation corresponds to an actual event record, so the challenges we encountered under right censoring are simply reversed.  A relevant quantity of interest in log-normal model applications is the mean, $\psi = \exp(\mu + \sigma^2/2)$, a non-linear function of $(\mu,\sigma)$.  Figure~\ref{fig:logsim} shows that, under various censoring levels, our proposed method gives marginal plausibility intervals for $\psi$ that achieve the nominal $100(1-\alpha)$\% coverage while, again, the other methods drastically under-cover.


\begin{figure}[p]
   \centering
   \subfloat[][$n=15$]{\includegraphics[width=.45\textwidth]{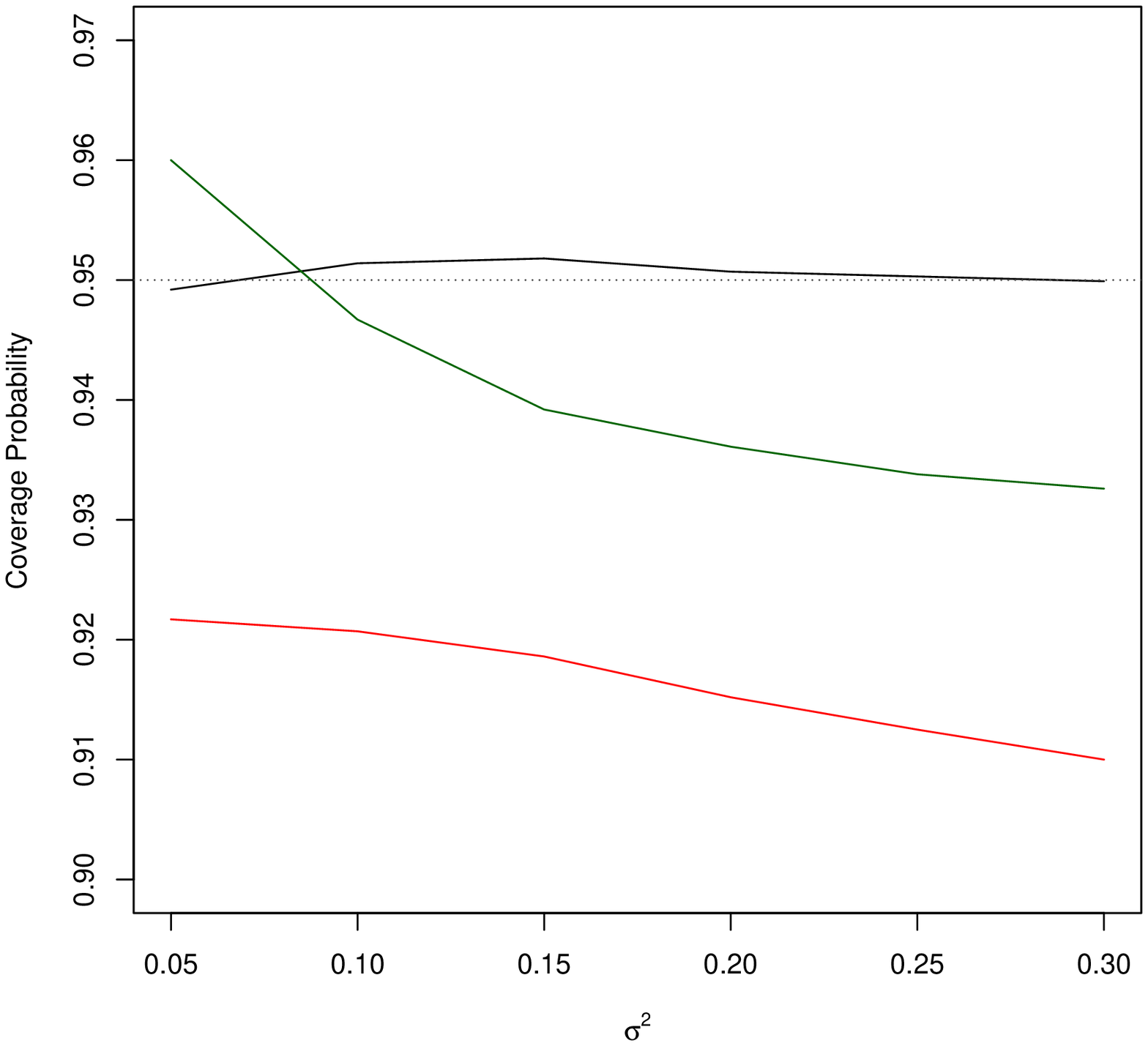}}\quad
   \subfloat[][$n=20$]{\includegraphics[width=.45\textwidth]{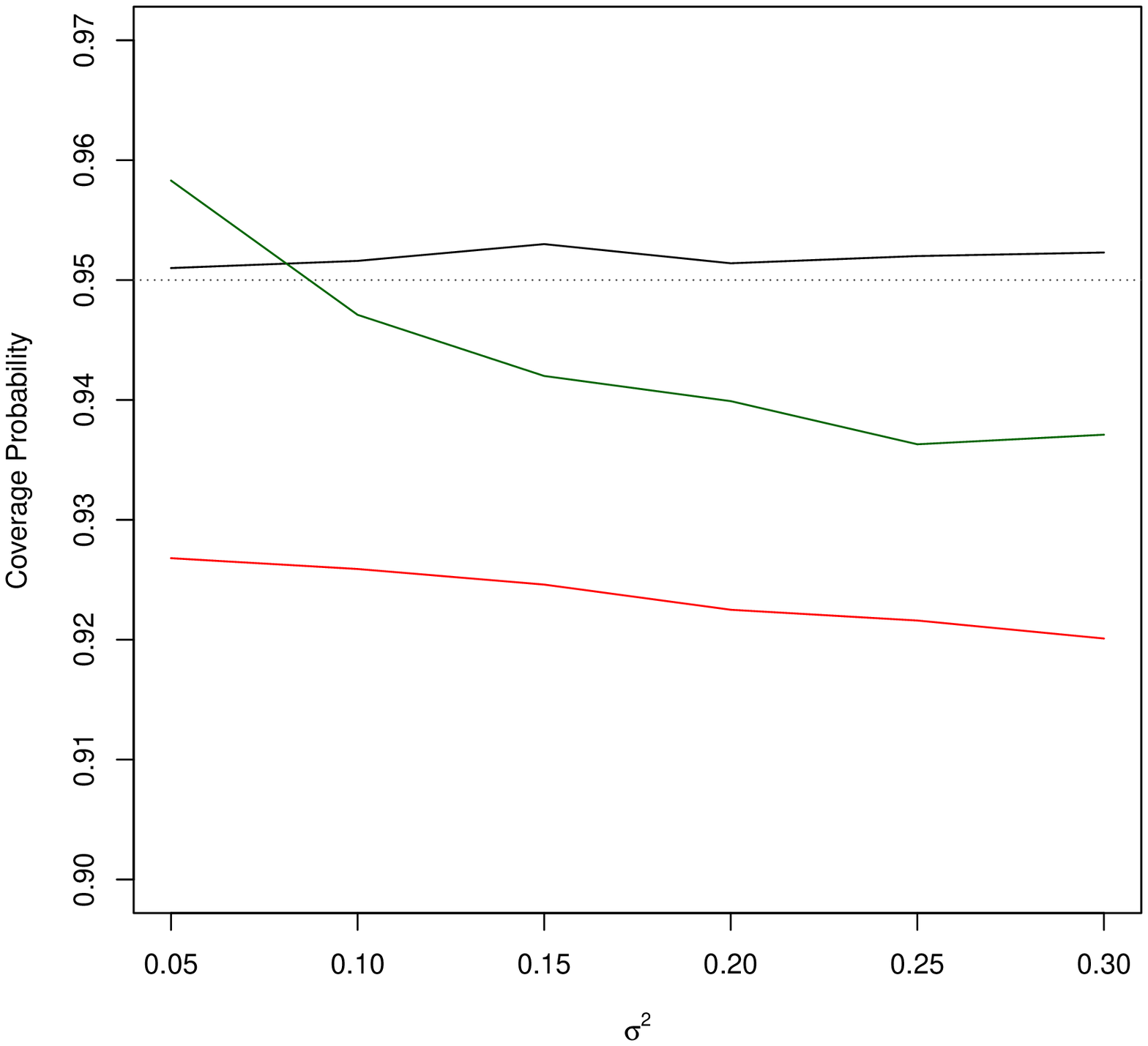}}\\
   \subfloat[][$n=25$]{\includegraphics[width=.45\textwidth]{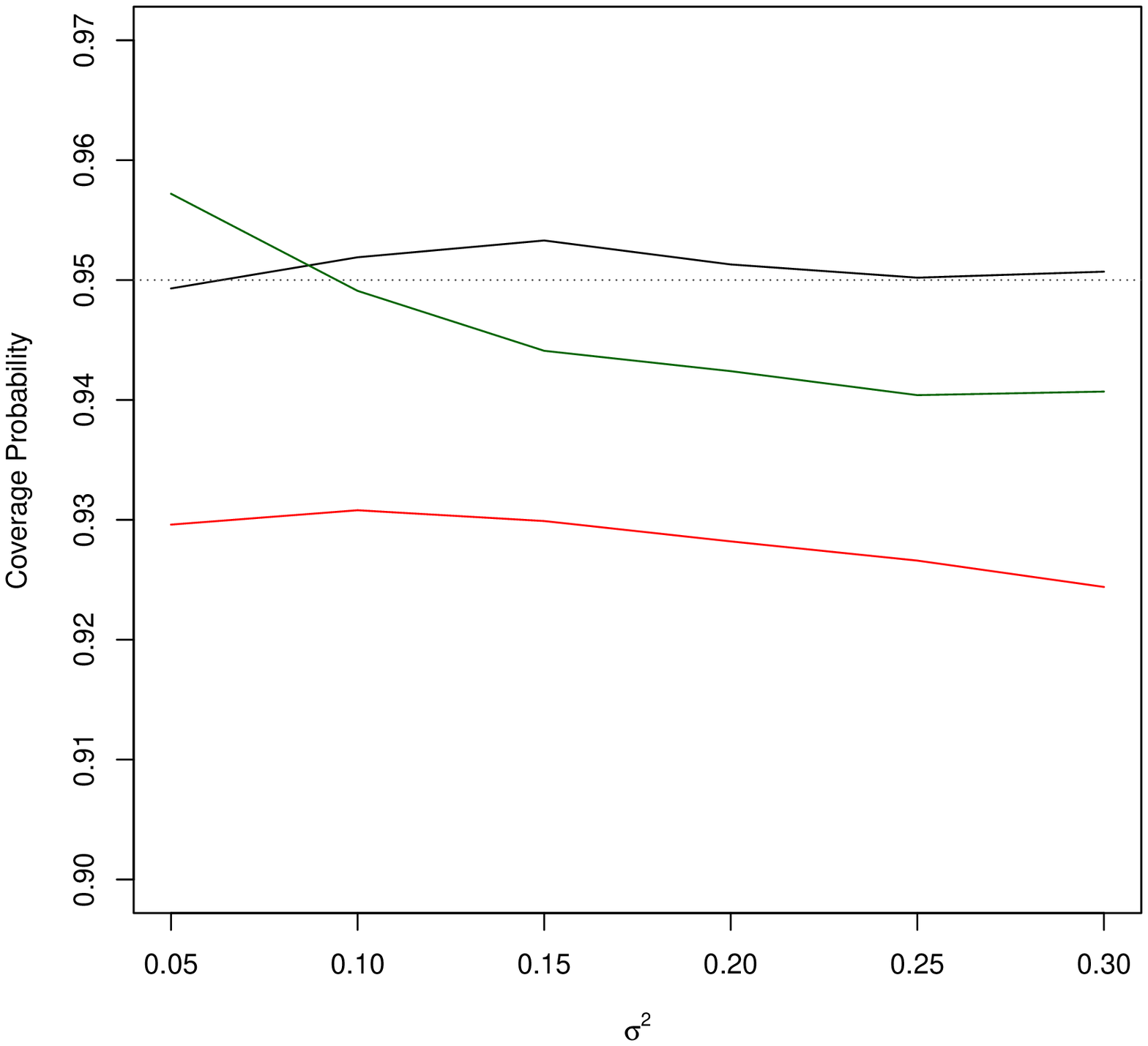}}\quad
   \subfloat[][$n=50$]{\includegraphics[width=.45\textwidth]{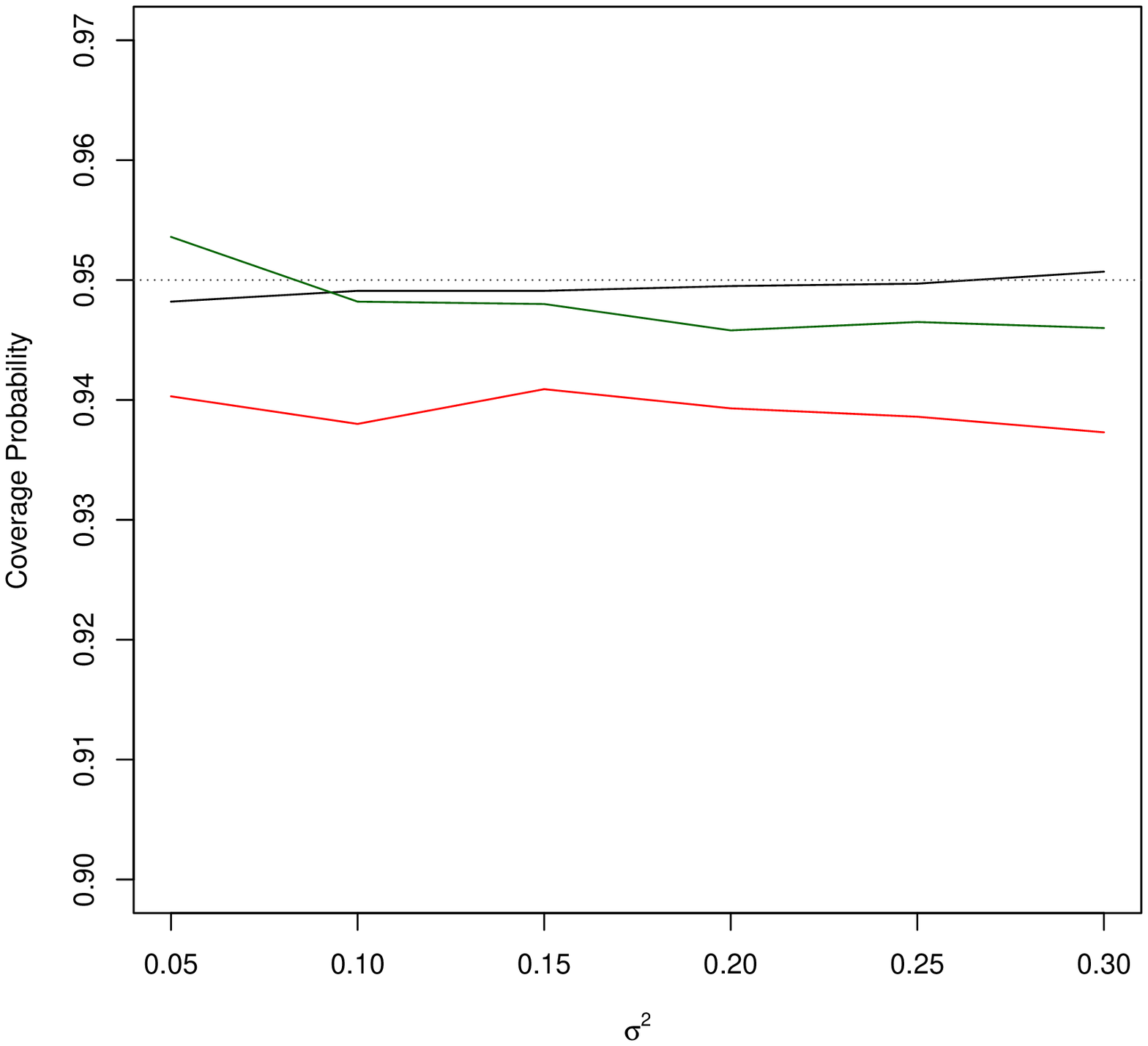}}
   \caption{Coverage probability of the 95\% plausibility interval for $\psi$ in the log-normal model (black). Results are compared to maximum likelihood (red) and Bayesian intervals based on a $\gam(1,0.1)$ prior on the precision $\tau = \sigma^{-2}$ and $\nm(0, 1000/\tau)$ prior on the mean (green).}
   \label{fig:logsim}
\end{figure}

We use Atrazine concentration data collected from a well in Nebraska as an example.  This set of 24 observations were randomly subject to two lower detection limits of 0.01 and 0.05 $\mu$g/l of which 11 observations were censored.  Despite this censoring level of 45.8\%, previous studies indicate the log-normality assumption holds \citep{Helsel2005}.  We apply our Monte Carlo approach to determine the joint plausibility contours for $\theta=(\mu, \sigma^2)$ in Figure \ref{fig:nebraska_joint}, along with the marginal plausibility function for the log-normal mean, $\psi$, in Figure \ref{fig:nebraska_mpl}.  The point at which we assign the highest plausibility aligns with the maximum likelihood estimator, $\hat{\mu} = -4.206$ and $\hat{\sigma} = 1.462$ \citep{Krish2011}.

\begin{figure}[t]
\centering
  \includegraphics[width=.75\textwidth]{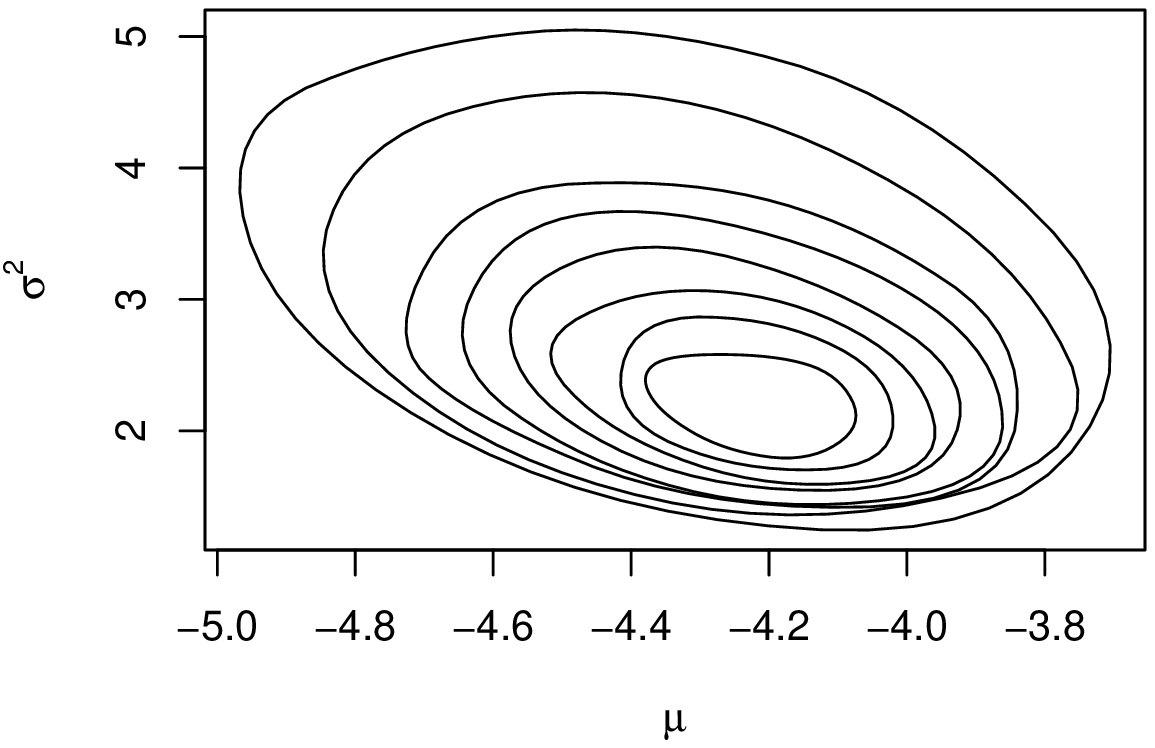}
  \caption{Plausibility contours at each $\alpha = 10$\% increment level beginning at 20\% for the Atrazine example under a log-normal model with Type I left censoring.}
  \label{fig:nebraska_joint}
\end{figure}

\begin{figure}[t]
\centering
  \includegraphics[width=.75\textwidth]{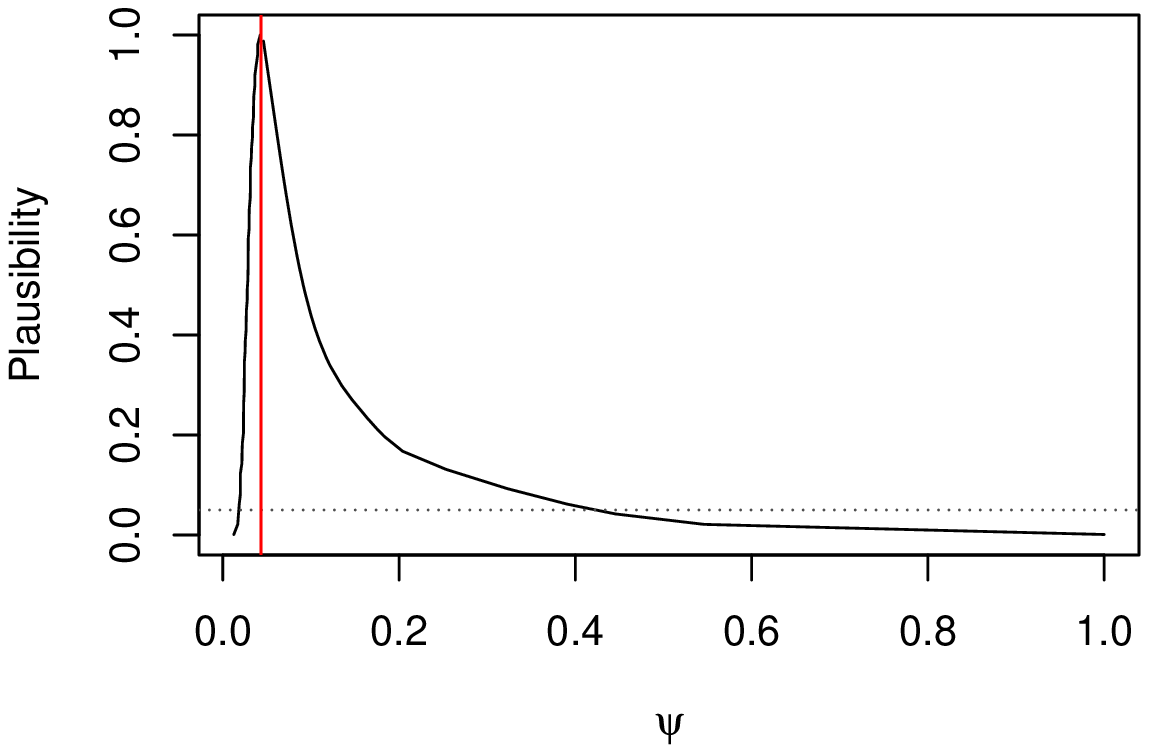}
  \caption{Marginal plausibility function for the mean $\psi$ in the Atrazine example. Reference lines at $\alpha = 0.05$\% and at the maximum likelihood estimate (red).}
  \label{fig:nebraska_mpl}
\end{figure}

\section{Conclusion}
\label{S:discuss}

In this paper, we proposed a specific inferential model construction for contexts in which the data are corrupted via censoring.  The main obstacle is that the censoring distribution is a unknown; despite not being of scientific interest, the presence of an infinite-dimensional nuisance parameter complicates the inferential model construction.  To overcome this challenge, we extend the generalized inferential model framework in \citet{gim} to cover the case of censoring according to a distribution $G$.  We propose a plug-in approximation to the known-$G$ inferential model construction with one that relies on a modified version of the classical Kaplan--Meier estimator, swapping the roles of event and censoring times.  Approximate validity is established in Theorem~\ref{thm:valid}, but we argued that the validity result is actually stronger than the theorem suggests.  We demonstrate numerically that the proposed inferential model approach outperforms traditional maximum likelihood and Bayesian solutions in terms of coverage probability.  


Aside from efforts to establish the validity property more rigorously for small $n$, it is of interest to explore complicated and practical types of censored-data models, e.g., ones where censoring depends on covariates so that an assumption of random censoring might not be warranted.  In principle, the approach described---with a generalized association based on the distribution of relative likelihood---would also work in more general cases, the optimization and Monte Carlo computations required to evaluate the distribution function $H_{\theta, G}$ would be much more involved.  Ongoing efforts are focused on this and other general improvements to the simple Monte Carlo computations described here.

\appendix

\section{Proof of Theorem~\ref{thm:valid}}
\label{SS:proof2}

Start by writing $p_{Y^n}(\theta; \hatG_n) = p_{Y^n}(\theta; G) + \Delta_n$, where 
\[\Delta_n = H_{\theta, \hatG_n}^n(R_{Y^n}(\theta)) - H_{\theta, G}^n(R_{Y^n}(\theta)), \] 
with $\hatG_n$ the Kaplan--Meier estimate and $G$ the true censoring distribution.  The key insight is that $p_{Y^n}(\theta; G)$ is exactly uniformly distributed under $\prob_{Y^n|\theta,G}$, so if we can show that $\Delta_n \to 0$ in probability, the claim will follow from Slutsky's theorem.  

A first observation is that  
\[ |\Delta_n| \leq \sup_{r \in [0,1]} \bigl| H_{\theta, \hatG_n}^n (r) - H_{\theta, G}^n(r) \bigr|, \]
so we can prove the claim by showing that the above difference vanishes uniformly.  But since these are distribution functions, it is enough to show that the difference vanishes pointwise, at each fixed $r$.  To prove pointwise convergence, we refer to \citet{banerjee2005} who shows that the usual large sample properties for the relative likelihood $R_{Y^n}(\theta)$ hold under $\prob_{Y^n| \theta, G}$ and under $\prob_{Y^n| \theta, G_n}$, as long as $G$ and $G_n$ are ``close.''  In particular, he shows that, for any $G_n$ that satisfies $G_n = G + n^{-1/2} Z_n$ for $Z_n$ bounded in probability, the two distributions $\prob_{Y^n| \theta, G}$ and $\prob_{Y^n| \theta, G_n}$ are mutually contiguous and, therefore, 
\begin{equation}
\label{eq:wilks}
-2\log R_{Y^n}(\theta) \to \chisq(\text{dim}(\theta)) \quad \text{in distribution as $n \to \infty$},
\end{equation}
under both $\prob_{Y^n| \theta, G}$ and $\prob_{Y^n| \theta, G_n}$; see also, \citet{murphy1997, murphy2000}.  Theorem 5 in \citet{Breslow1974} establishes that the Kaplan--Meier estimator satisfies 
\[ n^{1/2} \|\hatG_n - G\| = O(1) \quad \text{in probability as $n \to \infty$}, \]
where $\|G-G'\| = \sup_{t \leq \tau}|G(t)-G'(t)|$ and $\tau$ is any value such that $\{1-F_\theta(\tau)\}\{1-G(\tau)\} > 0$.  Therefore, we have 
\begin{equation}
H_{\theta, G}^n(r) \to H^\infty (r) \quad \text{and} \quad H_{\theta, \hatG_n}^n(r) \to H^\infty(r), \quad n \to \infty
\label{eq:conv}
\end{equation}
where $H^\infty$ is the limiting distribution function of $R_{Y^n}(\theta)$ from \eqref{eq:wilks}.  If we write 
\[ \bigl| H_{\theta, \hatG_n}^n(r) - H_{\theta, G}^n(r) \bigr| \leq \bigl| H_{\theta, \hatG_n}^n(r) - H^\infty (r)| +|H_{\theta, G}^n(r) - H^\infty (r) \bigr|, \] 
then we immediately see that the right-hand converges to 0 in $\prob_{Y^n|\theta,G}$-probability as $n \to \infty$.  This, in turn, implies the same for $\Delta_n$ and, applying Slutsky's theorem as discussed above, we can conclude that $p_{Y^n}(\theta; \hatG_n) \to \unif(0,1)$ in distribution.

\bibliographystyle{apa}
\bibliography{cahoon19.bib}

\end{document}